\documentclass[12pt,onecolumn,draft]{IEEEtran}
\usepackage{times,cite,lineno}
\usepackage{amsmath,amssymb,amsfonts,latexsym,epsf,enumerate,cases,algorithm,algorithmic,hhline,multirow,fancyhdr,subfigure}

\newtheorem{theorem}{Theorem}
\newtheorem{lemma}[theorem]{Lemma}

\begin{document}

\title{\huge{Weighted-Sum-Rate-Maximizing Linear Transceiver Filters for the K-User MIMO Interference Channel}}

\author{\large{Joonwoo~Shin and Jaekyun~Moon,~\IEEEmembership{Fellow,~IEEE}}


\thanks{This work was presented in part at IEEE Global Communications Conference 2011 and supported in part by the IT R\&D program of MKE/KEIT
(KI0038765,Development of B4G Mobile Communication Technologies
for Smart Mobile Services). The authors are with the School of
EECS, Korea Advanced Institute of Science and Technology (KAIST),
373-1,
Guseong-dong, Yuseong-gu, Daejeon, 305-701, Republic of Korea (e-mail: joonoos@etri.re.kr, jmoon@kaist.edu).}

}
%
%

\maketitle

\thispagestyle{empty}

\begin{abstract}

This letter is concerned with transmit and receive filter
optimization for the K-user MIMO interference channel.
Specifically, linear transmit and receive filter sets are designed
which maximize the weighted sum rate while allowing each
transmitter to utilize only the local channel state information.
Our approach is based on extending the existing method of
minimizing the weighted mean squared error (MSE) for the MIMO
broadcast channel to the K-user interference channel at hand. For
the case of the individual transmitter power constraint, however,
a straightforward generalization of the existing method does not reveal a viable solution.
It is in fact shown that there exists no closed-form solution for the transmit filter but
simple one-dimensional parameter search yields the desired solution.
Compared to the direct filter optimization using gradient-based
search, our solution requires considerably less computational
complexity and a smaller amount of feedback resources while
achieving essentially the same level of weighted sum rate.
A modified filter design is also presented which provides desired robustness in the presence of channel uncertainty.

\end{abstract}


\IEEEpeerreviewmaketitle

\setcounter{page}{0}
\newpage

\section{Introduction}

\PARstart{T}o achieve high spectral efficiency, much effort has been
focused on improving the achievable rate of multiple-input
multiple-output (MIMO) interference channels
\cite{Gomadam_IA:2008,LeeInKyu_KuserIC,Heath_IA:2011}. A notable scheme in this area, the interference
alignment (IA) technique of \cite{Jafar_IT} confines all undesired
interferences from other communication links into a pre-define
subspace and achieves a maximum-capacity scaling. However, it is also known that IA can only offer
a suboptimal sum rate at finite signal-to-noise ratios (SNRs) \cite{Heath_IA:2011}.

In this letter, we aim at maximizing the sum rate in the K-user
MIMO interference channel. We consider two linear transceiver
design methods. One is for the sum-power-usage-limit constraint
and the other applies to the per-transmit-node power-usage
constraint. The former can be viewed as a network-level constraint
whereas the latter is more of a device-level constraint. In both
designs, to maximize the weighted sum rate (WSR), we pursue
minimization of the weighted mean squared error (WMSE). The idea
of maximizing the WSR via receiver-side WMSE minimization was
originally developed for the multi-user MIMO broadcast channel
\cite{Cioffi:2008}. Our sum-power-constrained method could be seen
as a generalization of the approach of \cite{Cioffi:2008} to cover
the K-user MIMO interference channel and can be obtained as a
direct extension of the method in \cite{Cioffi:2008}. However, our
individual-power-constrained method is not a direct generalization
of the method of \cite{Cioffi:2008} due to multiple power
constraints. In fact, unlike in the case of the broadcast channel,
we show that there is no closed-form solution for the minimum WMSE
transmit filter, although a simple one-dimensional search for the
power-adjusting parameter leads to the desired solution. Using
simulation results and analysis, we verify that both proposed
schemes achieve the maximum WSR with lower computational
complexity than the gradient-based optimization of the transmit
and receive filters \cite{LeeInKyu_KuserIC}. Also, unlike in
\cite{LeeInKyu_KuserIC, Jafar_IT, ILee_SumMSE:2010}, our schemes
require only the local channel state information (CSI) (i.e., each
transmitter needs to know only the CSI of the links originating
from itself whereas the MIMO interference channel precoder designs
in \cite{LeeInKyu_KuserIC,Jafar_IT,ILee_SumMSE:2010} require the
CSI for all links). Additionally, we discuss modified transceiver
design that provides significant robustness in the presence of
inaccurate CSI.

Related ideas for the MIMO interference channel can also be found
in
\cite{Heath_IA:2011,Utschick:2009,Shen:2010,ILee_SumMSE:2010,Slock_PIMRC:2010,Shi_WMMSE_WSR}.
In \cite{Heath_IA:2011,Shen:2010}, the minimum MSE (MMSE)
transceiver is designed without considering different weights for
the MSEs at multiple receivers. In \cite{ILee_SumMSE:2010}
suboptimal MSE weights are used. In contrast, our weighted MMSE
transceiver design relies on a set of MSE weights that provides a
direct link between the weighted MMSE (WMMSE) and WSR criteria.
The WMMSE-based weighted utility maximization is also considered
in \cite{Utschick:2009}, but there only a single data stream is
assumed between a given user pair. A very similar idea on
maximizing WSR via WMSE minimization under the individual power
constraint has been discussed in \cite{Slock_PIMRC:2010}. But,
unlike in our approach, the inter-dependency between the
transmit-power-adjusting Lagrange multiplier and the precoding
matrix has not been considered in \cite{Slock_PIMRC:2010}. In our
individual-power-constrained transceiver design, this
inter-dependency is handled by introducing one-dimensional search
for the Lagrange multiplier. This means that the method of \cite{Slock_PIMRC:2010} requires recursive optimization
based on exchanges of filter-setting information among all transmitters.
Our method does not require recursive filter adjustment and no data exchanges are needed among
transmitters.\footnote{The independently conducted and recently published work of \cite{Shi_WMMSE_WSR}, which
was brought to our attention by an anonymous reviewer, also
pursues maximization of the WSR via weighted MSE
minimization. The transceivers in \cite{Shi_WMMSE_WSR} do become the
same as our proposed individual-power-constrained transceivers
when each base station serves a single user. Relative to the work of \cite{Shi_WMMSE_WSR},
this letter includes the sum-power-constrained method as well as a method to handle
mismatched CSI.}
Finally, we present a modified
transceiver design method for the imperfect-CSI environment and
analyze the computational complexity as well as the
required feedback amount in comparison with the gradient descent
method \cite{LeeInKyu_KuserIC}.

The following notations are used. We employ upper case boldface
letters for matrices and lower case boldface for vectors. For any
general matrix, ${\bf{X}}$, ${\bf{X}}^T$, ${\bf{X}}^*$,
${\bf{X}}^H$, $\text{Tr}({\bf{X}})$, $\text{det}({\bf{X}})$,
$\text{vec}(${\bf{X}}$)$,  $\text{SVD}(${\bf{X}}$)$ denote the
transpose, the conjugate, the Hermitian transpose, the trace, the
determinant, the stack columns, and the singular value
decomposition of ${\bf{X}}$, respectively. The symbol
${||}\cdot{||}_2$ indicates the 2-norm of a vector. The symbol ${\bf{I}}_n$
denotes an identity matrix of size $n$.

\section{System Model}

We consider the MIMO interference channel where precoding can only
be done over one transmission slot. As shown in Fig.
\ref{FIG:system_model}, $K$ source nodes simultaneously transmit
independent data streams to their desired destination nodes and
generate co-channel interference to all other undesired nodes. In
this system each source node $\lbrace\textsf{S}_k\rbrace$ is
equipped with $M$ antennas and each destination node
$\lbrace\textsf{D}_k\rbrace$ has $N$ antennas $(k\in\lbrace 1\sim
K\rbrace)$. The MIMO channels from $\textsf{S}_i$ to
$\textsf{D}_j$ are modelled by ${\bf{H}}_{ji}
\in{\mathcal{C}}^{N\times M}$ $(i,j\in\lbrace 1\sim K\rbrace)$
whose coefficients are independent and identically distributed
(i.i.d) complex Gaussian random variables with
$\mathcal{CN}(0,\sigma_h^2)$. We assume that the channel
information is only \emph{locally} available, i.e., each node
knows only the coefficients for the channel link originating from
itself. Note that the precoder designs of
\cite{LeeInKyu_KuserIC,Jafar_IT,ILee_SumMSE:2010} are based on the
availability of the \emph{global} channel information. Let
${\bf{s}}_k\in{\mathcal{C}}^{d\times 1}$ denote the symbol vector
from $\textsf{S}_k$ with $\mathbb{E}\lbrack
{\bf{s}}_k{\bf{s}}_k^H\rbrack = {\bf{I}}_d$ where $d$ is the
number of data streams for $\textsf{D}_k$, $d\le M,N$ and the
value of $d$ is chosen to meet the feasibility of degree of
freedom \cite{Bresler_Feasibility_IA}.  Also
${\bf{V}}_k\in{\mathcal{C}^{M\times d}}$ denotes the precoding
matrix for $\textsf{S}_k$.  Then, the $N\times 1$ received signal
vector at $\textsf{D}_k$ is represented as
\begin{equation}
{\bf{y}}_{k} = {\bf{H}}_{kk}{\bf{V}}_{k}{\bf{s}}_{k} + \sum_{i\ne
k }^{K}{\bf{H}}_{ki}{\bf{V}}_{i}{\bf{s}}_{i}+{\bf{n}}_{k},
\label{y_k}
\end{equation}
where ${\bf{n}}_{k}$ denotes the i.i.d complex Gaussian noise
vector at $\textsf{D}_k$ with
$\mathcal{CN}({\bf{0}},\sigma_n^2{\bf{I}}_N)$. Then, $\textsf{D}_k$
combines its received signal with
${\bf{U}}_k\in{\mathcal{C}^{d\times N}}$ to decode the desired
signals:
\begin{equation}
\hat{\bf{s}}_{k} = {\bf{U}}_{k}{\bf{y}}_{k} =
{\bf{U}}_{k}{\bf{H}}_{kk}{\bf{V}}_{k}{\bf{s}}_{k} +
{\bf{U}}_{k}\sum_{i\ne k
}^{K}{\bf{H}}_{ki}{\bf{V}}_{i}{\bf{s}}_{i}+{\bf{U}}_{k}{\bf{n}}_{k}.
\label{s_k_hat}
\end{equation}
Our goal is to find $\lbrace {\bf{V}}_k \rbrace$ and $\lbrace
{\bf{U}}_k \rbrace$ that maximize the WSR under the sum-power
constraint and also the individual-power constraint. We assume a
unit noise variance ($\sigma_n^2=1$) without losing generality.
\section{Weighted Sum Rate Maximization}
First consider finding $\lbrace {\bf{V}}_k \rbrace$ that maximizes
\begin{equation}
\sum_{k=1}^{K}\mu_k R_k  \quad \text{subject to }\sum_k
\text{Tr}({\bf{V}}_k{\bf{V}}_k^H)= P_T \quad \text{or} \quad
\text{Tr}({\bf{V}}_k{\bf{V}}_k^H)= P_k \quad \forall k \label{WSR}
\end{equation}
where the subscript $k$ points the source node and its intended
destination node, $\mu_k$ denotes the weight, $R_k$ is the
achievable rate, $P_T$ represents the maximum sum power allowed
for all transmitters and $P_k$ is the $k$-th node's maximum
transmit power. With Gaussian signaling, the achievable rate takes the
well-known form:
\begin{equation}
 \quad R_k= \text{log } \Big{\lbrace}\text{det}\Big{(} {\bf{I}}_N +
{\bf{\Phi}}_{k}^{-1}{\bf{H}}_{kk}{\bf{V}}_{k}{\bf{V}}_{k}^H{\bf{H}}_{kk}^H
\Big{)}\Big{\rbrace},  \label{R_k}
\end{equation}
where ${\bf{\Phi}}_{k} = {\bf{I}}_N + \sum_{i\ne k}^K
{\bf{H}}_{ki}{\bf{V}}_{i}{\bf{V}}_{i}^H{\bf{H}}_{ki}^H$. We attempt to solve this WSR maximization problem
by minimizing the weighted receiver MSE, as
has been done for the MIMO broadcast channel \cite{Cioffi:2008}.
This approach was also attempted for the K-user MIMO interference channel in \cite{Slock_PIMRC:2010}
under the individual-power constraint,
but our solution is different as elaborated below.

\subsection{Relationship between achievable rate and error covariance matrix}
To understand the link between the WSR maximization problem and
the WMSE minimization problem in the K-user MIMO interference
channel, we need to clarify the relationship between the
achievable rate and the error covariance matrix. This argument is
parallel to one given in \cite{Cioffi:2008} for the MIMO broadcast
channel. For the MMSE receive filter at $\textsf{D}_k$, we write
\begin{align}
{\bf{U}}_{k}^{(MMSE)} = &\arg\min
\mathbb{E}||{\bf{U}}_k{\bf{y}}_k-{\bf{s}}_k ||_2^2 \nonumber\\
=&{\bf{V}}_{k}^{H}{\bf{H}}_{kk}^{H}(\sum_{i=1}^{K}{\bf{H}}_{ki}{\bf{V}}_{i}{\bf{V}}_{i}^{H}{\bf{H}}_{ki}^{H}+{\bf{I}}_N)^{-1},
\label{U_k}
\end{align}
and the error matrix for $\textsf{D}_k$ is given by
\begin{align}
{\bf{E}}_{k}= &\mathbb{E}\lbrace({\bf{U}}_k^{(MMSE)}{\bf{y}}_k-{\bf{s}}_k )({\bf{U}}_k^{(MMSE)}{\bf{y}}_k-{\bf{s}}_k )^{H}\rbrace \nonumber\\
=&({\bf{I}}_N +
{\bf{\Phi}}_{k}^{-1}{\bf{H}}_{kk}{\bf{V}}_{k}{\bf{V}}_{k}^H{\bf{H}}_{kk}^H)^{-1}.
\label{Emtx}
\end{align}
Comparing (\ref{R_k}) and (\ref{Emtx}), the relationship between the
achievable rate and the error covariance matrix is established as:
\begin{equation}
R_k = \text{log }\lbrace\text{det}({\bf{E}}_k^{-1})\rbrace
\label{R_k_E_k}
\end{equation}
which, not surprisingly, is identical to the relationship
between the rate and the error covariance matrix for the case
of the MIMO broadcast channel \cite{Cioffi:2008}. Apparently, though,
the error covariance matrix ${\bf{E}}_{k}$ here is different from that of the broadcast channel
due to the presence of multiple sources.
Note that
this relationship between the achievable rate and the error
covariance matrix holds for any $\lbrace {\bf{V}}_k \rbrace$,
implying that (\ref{R_k_E_k}) is true with either transmit
power constraint.

\subsection{MSE weight design}

Now consider finding $\lbrace
{\bf{V}}_k \rbrace$ that solves the
following WMMSE problem:
\begin{equation}
\min \sum_{k=1}^{K} \text{Tr} ({\bf{W}}_k {\bf{E}}_k) \quad
\text{subject to }\sum_k \text{Tr}({\bf{V}}_k{\bf{V}}_k^H)= P_T
\quad \text{or} \quad \text{Tr}({\bf{V}}_k{\bf{V}}_k^H)= P_k \quad
\forall k ,\label{WMMSE}
\end{equation}
where ${\bf{W}}_k\in\mathcal{C}^{d\times d}$ represents the MSE
weight. Again following the argument of \cite{Cioffi:2008}, the
MSE weights can be chosen so that both WSR and WMMSE problems have
a common solution. For this, set up the Lagrangians for
(\ref{WSR}) and (\ref{WMMSE}):
\begin{equation}
\mathcal{L}_{WSR} = -\sum_{k=1}^{K}\mu_k R_k
+\theta\lambda(\sum_{k=1}^{K}\text{Tr}({\bf{V}}_k{\bf{V}}_k^H)-P_T)+(1-\theta)\Big{(}\sum_{k=1}^{K}{\lambda}_k(\text{Tr}({\bf{V}}_k{\bf{V}}_k^H)-P_k)\Big{)}
\nonumber
\end{equation}
and
\begin{equation}
\mathcal{L}_{WMSE}=
\sum_{k=1}^{K}\text{Tr}({\bf{W}}_k{\bf{E}}_k)+\theta\lambda(\sum_{k=1}^{K}\text{Tr}({\bf{V}}_k{\bf{V}}_k^H)-P_T)+(1-\theta)\Big{(}\sum_{k=1}^{K}{\lambda}_k(\text{Tr}({\bf{V}}_k{\bf{V}}_k^H)-P_k)
\Big{)} \nonumber
\end{equation}
respectively, where $\theta$ selects the desired power constraint
('$\theta =1$' for the sum power constraint and '$\theta=0$' for
the individual power constraint), $\lambda$ and $\lbrace \lambda_k
\rbrace$ denote the Lagrange multipliers for the two transmit
power constraints. Next, equate their gradients obtained via the
matrix derivative formulas:
$d\lbrace\text{ln}(\text{det}({\bf{X}}))\rbrace=\text{Tr}\lbrace{\bf{X}}^{-1}d({\bf{X}})\rbrace$,
$d\lbrace\text{Tr}({\bf{X}})\rbrace=\text{Tr}\lbrace
d({\bf{X}})\rbrace$, $\text{vec}\lbrace
d({\bf{X}})\rbrace=d\lbrace\text{vec}({\bf{X}})\rbrace$,
$\text{Tr}({\bf{X}}^T{\bf{Y}})=\text{vec}({\bf{X}})^T\text{vec}({\bf{Y}})$.
Subsequently, the resulting MSE weight can be found as
\begin{equation}
{\bf{W}}_{k}=\frac{\mu_k}{\text{ln(2)}}{\bf{E}}_{k}^{-1}.\label{W_opt}
\end{equation}
Note that the choice of the MSE weights $\lbrace {\bf{W}}_k \rbrace$ is
irrelevant to the transmit power constraint, which makes sense as
$\lbrace {\bf{W}}_k \rbrace$ are receiver-side design
parameters.

\subsection{Sum power constrained precoder design}
We are now ready to find the transmit precoding matrix that
minimizes the WMSE under the sum-power constraint, i.e., find $\lbrace {\bf{V}}_k \rbrace$ that minimizes
\begin{equation}
\sum_{k=1}^{K} \mathbb{E} \lbrack \text{Tr}\lbrace {\bf{W}}_{k}
({\bf{s}}_k -\beta^{-1}\hat{\bf{s}}_k )({\bf{s}}_k
-\beta^{-1}\hat{\bf{s}}_k )^H\rbrace \rbrack  \quad \text{subject
to } \sum_{k}\text{Tr}({\bf{V}}_k{\bf{V}}_k^H)=
P_T\label{WMMSE_Vk}
\end{equation}
where $\lbrace {\bf{W}}_k \rbrace$ is set according to (\ref{W_opt}) and
$\beta$ is a scaling parameter. With matrix derivative
formulas, the WMMSE transmit filter that satisfies (\ref{WMMSE_Vk}) can be shown to be
\begin{align}
{\bf{V}}_{k} = \beta {\bf{V}}_{k}^{'},\label{V_k}
\end{align}
 where ${\bf{V}}_{k}^{'}=\Big{(}{\bf{\Psi}}_{k}+
\frac{\sum_{i}\text{Tr}( {\bf{W}}_{i}
{\bf{U}}_{i}{\bf{U}}_{i}^{H})}{P_{T}}{\bf{I}}_{M}\Big{)}^{-1}
{\bf{H}}_{kk}^{H}{\bf{U}}_{k}^{H}{\bf{W}}_{k}$,
 ${\bf{\Psi}}_{k}=\sum_{i=
1}^{K}{\bf{H}}_{ik}^{H}{\bf{U}}_{i}^{H}{\bf{W}}_{i}{\bf{U}}_{i}{\bf{H}}_{ik}
$, and $\beta=\sqrt{\frac{P_T}{ \sum_k\text{Tr}
({\bf{{V}}}_{k}^{'} {{\bf{{V}}}_{k}^{'}}^H )}}$.

This result is a rather straightforward generalization of the WMMSE
precoder in the broadcast channel. It can indeed be seen that setting
${\bf{H}}_{ki}={\bf{H}}_{kk}$ for all $i$, our solutions
(\ref{U_k}), (\ref{W_opt}), and (\ref{V_k}) reduce to the
respective receive filter, MSE weight and transmit filter
solutions obtained for the multi-user MIMO broadcast channels
through WMSE minimization \cite{Cioffi:2008}.

\subsection{Individual power constrained transceiver design}

Now let us consider the individual-power-constrained network. We proceed to find the transmit filter
that minimizes the weighted MSE:
\begin{align}
\sum_{k=1}^{K} \mathbb{E} \lbrack \text{Tr}\lbrace {\bf{W}}_{k}
({\bf{s}}_k -\hat{\bf{s}}_k )({\bf{s}}_k -\hat{\bf{s}}_k
)^H\rbrace \rbrack \quad\text{subject to }
\text{Tr}({\bf{V}}_k{\bf{V}}_k^H)= P_k \quad \forall k.
\label{WMMSE_Vk_indP}
\end{align}
Again equating the gradients of the Lagrangians corresponding to
the WMMSE and WSR maximization procedures and using the matrix
derivative formulas, the WMMSE transmit filter at $\textsf{S}_k$
is found as:
\begin{align}
{\bf{{V}}}_{k} =& \Big{(}{\bf{\Psi}}_{k} + \lambda_k
{\bf{I}}_{M}\Big{)}^{-1}
{\bf{H}}_{kk}^{H}{\bf{U}}_{k}^{H}{\bf{W}}_{k}\label{V_k_indP}
\end{align}
where $\lambda_k $ is set to satisfy the transmit power
constraint at $\textsf{S}_k$ and again $\lbrace {\bf{W}}_k \rbrace$ are as given in (\ref{W_opt}).
Unlike the sum-power-constrained
WMMSE precoders of (\ref{V_k}), for which the power control parameters
are found in closed form, here we resort to a numerical method to find $\lambda_k$,
due to the inter-dependency between ${\bf{{V}}}_{k}$ and
$\lambda_k$ in (\ref{V_k_indP}). Fortunately, based on the
following lemma, $\lambda_k$ can be found with simple
one-dimensional (1-D) numerical search.
\begin{lemma}\label{Lem:lambda}
The per-node transmit power,
$\text{Tr}({\bf{V}}_k(\lambda_k){\bf{V}}_k(\lambda_k)^H)$, is a
monotonically decreasing function of $\lambda_k$.
\end{lemma}
\begin{proof}
Let
$\text{SVD}({\bf{\Psi}}_{k})={\bf{Q}}_{k}{\bf{\Sigma}}_k{\bf{Q}}_{k}^{H}$.
Then, the transmit power at $\textsf{S}_k$ is given by
\begin{align}
\text{Tr}\lbrace({\bf{V}}_k(\lambda_k){\bf{V}}_k(\lambda_k)^H\rbrace=&\text{Tr}\lbrace({\bf{Q}}_{k}{\bf{\Sigma}}_k{\bf{Q}}_{k}^{H}+\lambda_k{\bf{I}_M})^{-1}{\bf{H}}_{kk}^H{\bf{U}}_{k}^H{\bf{W}}_{k}{\bf{U}}_{k}{\bf{H}}_{kk}({\bf{Q}}_{k}{\bf{\Sigma}}_k{\bf{Q}}_{k}^{H}+\lambda_k{\bf{I}_M})^{-1}\rbrace\nonumber
\\
=&\text{Tr}\lbrace({\bf{\Sigma}}_k+\lambda_k{\bf{I}_M})^{-2}{\bf{Q}}_{k}^H{\bf{H}}_{kk}^H{\bf{U}}_{k}^H{\bf{W}}_{k}{\bf{U}}_{k}{\bf{H}}_{kk}{\bf{Q}}_{k}\rbrace\nonumber
\\
 =& \sum_{i=1}^{K}\frac{\lbrack
{\bf{\Pi}}_k\rbrack_{i,i}}{(\sigma_{k,i}+\lambda_k)^2},
\label{App:Lambda}
\end{align}
where
${\bf{\Pi}}_k={\bf{Q}}_{k}^H{\bf{H}}_{kk}^H{\bf{U}}_{k}^H{\bf{W}}_{k}{\bf{U}}_{k}{\bf{H}}_{kk}{\bf{Q}}_{k}$
, $\lbrack {\bf{\Pi}}_k\rbrack_{i,i}$ is the $ii$-th element of
${\bf{\Pi}}_k$, and $\sigma_{k,i}$ is the $i$-th element of
${\bf{\Sigma}}_k$. Because $\lambda_k\ge 0$,
$\text{Tr}\lbrace{\bf{V}}_k(\lambda_k){\bf{V}}_k(\lambda_k)^H\rbrace$
is monotonically decreasing with $\lambda_k$.
\end{proof}

Note that the proper set of MSE weights for the K-user MIMO
interference channel has already been derived in
\cite{Slock_PIMRC:2010} in the process of establishing a
connection between the WMMSE problem and the WSR maximization
problem. In \cite{Slock_PIMRC:2010}, though, the transmitter
$\textsf{S}_k$ is expressed as a function of
itself as well as transmitters at the other nodes, i.e.
${\bf{V}}_{k} =f(\lbrace {\bf{V}}_{1},\cdots, {\bf{V}}_{K}
\rbrace)$. The consequence of this formulation is that the transmitter solution in \cite{Slock_PIMRC:2010}
cannot be found without recursive calculation and additional
filter-setting information exchanges among all transmit nodes.
In contrast, our transmit
filter design is based on a clear recognition of the inter-dependency between
$\lambda_k$ and ${\bf{V}}_k$, and as a result the proposed transmit
filter (\ref{V_k_indP}) can be found through a simple 1-D numerical search
with no additional information ($\lbrace {\bf{V}}_l \rbrace$
($l\ne k$)) exchanges needed among the transmit nodes.

\subsection{Iterative algorithm to maximize the weighted sum rate}

In the previous sections, we found the MSE weights and then subsequently
WMMSE receive and transmit filters with both the sum-power
constraint and the individual-power constraint. Each of three
sets of parameters - MSE weights, transmit filters and receive filters -
is derived assuming the other sets are given.
In practice, to find optimum WSR solutions, the
inter-dependencies between the parameters are handled with
the following iterative or alternating optimization algorithm.
\begin{algorithm}
\caption{Obtaining the optimal WSR transceivers via the WMMSE
criterion} \label{Alg:proposed}
\begin{algorithmic}
\STATE{Initialize $l=0$ and $\lbrace{\bf{V}}_k^{(0)}\rbrace$,
calculate  ${R}_{sum}^{(0)}$}.

\REPEAT

\STATE{$l:=l+1$}

\STATE{Step 1: Calculate ${\bf{U}}_k^{(l)}|\lbrace
{\bf{V}}_i^{(l-1)}\rbrace$ for all $k$ using (\ref{U_k}).}

\STATE{Step 2: Calculate ${\bf{W}}_k^{(l)}|\lbrace
{\bf{V}}_i^{(l-1)}\rbrace$ for all $k$ using (\ref{W_opt}).}

\STATE{Step 3: Calculate ${\bf{V}}_k^{(l)}|\lbrace
{\bf{U}}_i^{(l)}\rbrace$, $\lbrace {\bf{W}}_i^{(l)}\rbrace$ for
all $k$ using (\ref{V_k}) for the sum power constrained case or
(\ref{V_k_indP}) for the individual power constrained case.}

\UNTIL{$|{R}_{sum}^{(l)}-{R}_{sum}^{(l-1)}|<\epsilon$, where
$\epsilon$ is some arbitrarily small value and $R_{sum}=\sum_k \mu_k
R_k$.}
\end{algorithmic}
\end{algorithm}

The algorithm is common to both the sum-power-constrained design
and the individual-power-constrained design. This algorithm is
provably convergent to a local optimum; this can be shown by
proving monotonic convergence of an equivalent optimization
problem based on expanding the WSR maximization problem of
(\ref{WSR}) to add the MMSE weights and receive filters as
optimization variables, as has been done for the MIMO broadcast
channel in \cite{Cioffi:2008}. We note, however, that this
algorithm does not guarantee the global optimal solution, since
the WMMSE minimization (\ref{WMMSE}) is not \emph{jointly} convex
over all input variables. To reasonably approach the optimal
solution one must resort to repeated runs of the algorithm using
different initial settings, or, for computationally efficient initialization,
choose $\lbrace {\bf{V}}_{k}^{(0)}\rbrace$ in
Step 1 from the right singular matrices of $\lbrace
{\bf{H}}_{kk}\rbrace$ or from random matrices generated
according to the normal distribution with zero mean and unit variance
\cite{Shen:2010}.

\section{Robust transceiver design for imperfect channel information}
 In practical scenarios, mismatch between the true channel
$\lbrace{\bf{H}}_{ij}\rbrace$ and the estimated channel (denoted by
$\lbrace\tilde{\bf{H}}_{ij}\rbrace$) is inevitable because of the
channel estimation errors \cite{Tresch_CHE_IC:2010}. In this section, we
design robust transceivers for mitigating the performance
degradation caused by channel mismatch.  We assume that $\lbrace
\tilde{\bf{H}}_{ij}\rbrace$ is related to
$\lbrace{\bf{H}}_{ij}\rbrace$ by
$\tilde{\bf{H}}_{ij}={\bf{H}}_{ij}+{\bf{\Delta}}_{ij}$ where the
elements of ${\bf{\Delta}}_{ij}$ are independent and identically distributed (i.i.d.) complex Gaussian
random variables with variance $\sigma_{\Delta}^2$
\cite{Tresch_CHE_IC:2010}. Then, the received signal can be
rewritten as
\begin{equation}
{\tilde{{\bf{s}}}}_{k} =
{\tilde{\bf{U}}}_{k}({\tilde{\bf{H}}}_{kk}-{\bf{\Delta}}_{kk}){\tilde{\bf{V}}}_{k}{\bf{s}}_{k}
+ {\tilde{\bf{U}}}_{k}\sum_{i\ne k
}^{K}({\tilde{\bf{H}}}_{ki}-{\bf{\Delta}}_{ki}){\tilde{\bf{V}}}_{i}{{\bf{s}}}_{i}+{\tilde{\bf{U}}}_{k}{\bf{n}}_{k}
\label{un_s_k_hat}
\end{equation}
where $\lbrace{\tilde{\bf{V}}}_{k}\rbrace$ and
$\lbrace{\tilde{\bf{U}}}_{k}\rbrace$ are computed from $\lbrace
{\tilde{\bf{H}}}_{ij}\rbrace$ with no knowledge of the presence of
$\lbrace{\bf{\Delta}}_{ij}\rbrace$. We try to mitigate the
effect of channel mismatch by minimizing the appropriate metrics
averaged over ${\bf{\Delta}}_{ij}$'s.

\subsubsection{Modified MSE weight} Following the design procedure in previous
sections, a modified version of the MMSE receiver filter is found as
${\tilde{\bf{U}}}_{k}
={\tilde{\bf{V}}}_{k}^{H}{\tilde{\bf{H}}}_{kk}^{H}(\sum_{i=1}^{K}{\tilde{\bf{H}}}_{ki}{\tilde{\bf{V}}}_{i}{\tilde{\bf{V}}}_{i}^{H}{\tilde{\bf{H}}}_{ki}^{H}+\sum_{i=1}^{K}
\sigma_{\Delta}^2\text{Tr}({\bf{\Lambda}}_{{\tilde{\bf{V}}}_i}){\bf{I}}_N+{\bf{I}}_N)^{-1}$,
 where $\text{SVD}({\tilde{\bf{V}}}_i{\tilde{\bf{V}}}_i^H)={\bf{Q}}_i{\bf{\Lambda}}_{{\tilde{\bf{V}}}_i}{\bf{Q}}_i^H$.
The modified MSE weights that force the optimum solutions of
the WSR maximization and WMMSE problems to be identical are derived as
${\tilde{\bf{W}}}_k =
\frac{\mu_k}{\text{ln}(2)}{\tilde{\bf{E}}}_k^{-1}$, where
${\tilde{\bf{E}}}_k=({\bf{I}}_N+
{\tilde{\bf{\Phi}}}_{k}^{-1}{\tilde{\bf{H}}}_{kk}{\tilde{\bf{V}}}_{k}{\tilde{\bf{V}}}_{k}^H{\tilde{\bf{H}}}_{kk}^H)^{-1}$
and ${\tilde{\bf{\Phi}}_{k}}={\bf{I}}_N + \sum_{i\ne k}^K
{\tilde{\bf{H}}}_{ki}{\tilde{\bf{V}}}_{i}{\tilde{\bf{V}}}_{i}^H{\tilde{\bf{H}}}_{ki}^H
+ \sum_{i=1}^{K}
\sigma_{\Delta}^2\text{Tr}({\bf{\Lambda}}_{{\tilde{\bf{V}}}_i}){\bf{I}}_N$.

\subsubsection{Robust transceiver design with the sum power constraint}
The modified transmit filters are derived based on the following
optimization problem:
\begin{align}
\min \sum_{k=1}^{K} \mathbb{E} \lbrack \text{Tr}\lbrace
{\tilde{\bf{W}}}_{k} ({\bf{s}}_k
-{\tilde\beta}^{-1}{\tilde{{\bf{s}}}}_k )({\bf{s}}_k
-{\tilde\beta}^{-1}{\tilde{{\bf{s}}}}_k )^H\rbrace \rbrack \quad
\text{subject to }
\sum_{k}\text{Tr}({\tilde{\bf{V}}}_k{\tilde{\bf{V}}}_k^H)= P_T
\label{Mod_WMMSE_Vk}.
\end{align} Utilizing matrix derivative formulas, the resultant modified-WMMSE transmit filters
are obtained as
\begin{equation}
{\tilde{\bf{V}}}_{k} ={\tilde\beta}{\tilde{\bf{V}}}_{k}^{'}
\label{un_V_k1},
\end{equation}
where ${\tilde{\bf{V}}}_{k}^{'}=\Big{(}{\tilde{\bf{\Psi}}}_{k} +
\frac{\sum_{i=1}^{K}\text{Tr}( {\tilde{\bf{W}}}_{i}
{\tilde{\bf{U}}}_{i}{\tilde{\bf{U}}}_{i}^{H})}{P_{T}}{\bf{I}}_{M}
+ \sum_{i=1}^{K} \sigma_{\Delta}^2
\text{Tr}({\bf{\Lambda}}_{{\tilde{\bf{U}}}_i}){\bf{I}}_M
\Big{)}^{-1}
{\tilde{\bf{H}}}_{kk}^{H}{\tilde{\bf{U}}}_{k}^{H}{\tilde{\bf{W}}}_{k}$,
 ${\tilde\beta} =\sqrt{\frac{P_T}{
\sum_k\text{Tr}({\tilde{\bf{V}}}_{k}^{'}
{{\tilde{\bf{V}}}_{k}^{'H}} )}}$,
${\tilde{\bf{\Psi}}}_{k}=\sum_{i=
1}^{K}{\tilde{\bf{H}}}_{ik}^{H}{\tilde{\bf{U}}}_{i}^{H}{\tilde{\bf{W}}}_{i}{\tilde{\bf{U}}}_{i}{\tilde{\bf{H}}}_{ik}$,
 and
$\text{SVD}({\tilde{\bf{U}}}_{i}^{H}{\tilde{\bf{W}}}_{i}
{\tilde{\bf{U}}}_{i})={\bf{Q}}_i^{'}{\bf{\Lambda}}_{{\tilde{\bf{U}}}_i}{\bf{Q}}_i^{'H}$.

\subsubsection{Robust transceiver design with the individual power constraint}
 The optimization problem to derive the modified precoder is
\begin{align}
\min \sum_{k=1}^{K} \mathbb{E} \lbrack \text{Tr}\lbrace
{\tilde{\bf{W}}}_{k} ({\bf{s}}_k -{\tilde{{\bf{s}}}}_k
)({\bf{s}}_k -{\tilde{{\bf{s}}}}_k )^H\rbrace \rbrack \quad
\text{subject to  }
\text{Tr}({\tilde{\bf{V}}}_k{\tilde{\bf{V}}}_k^H)= P_k
\quad\forall k \label{Mod_WMMSE_Vk_indP}.
\end{align}
With the matrix derivative formulas, the modified-WMMSE
transmit precoder at $\textsf{S}_k$ with the individual power
constraint is written as
\begin{align}
{\tilde{\bf{{V}}}}_{k} =& \Big{(}{\tilde{\bf{\Psi}}}_{k} +
{\tilde{\lambda}}_k{\bf{I}}_{M} + \sum_{i=1}^{K} \sigma_{\Delta}^2
\text{Tr}({\bf{\Lambda}}_{{\tilde{\bf{U}}}_i}){\bf{I}}_M
\Big{)}^{-1}
{\tilde{\bf{H}}}_{kk}^{H}{\tilde{\bf{U}}}_{k}^{H}{\tilde{\bf{W}}}_{k}\label{un_V_k2_indP}
\end{align}
where the power control parameter ${\tilde{\lambda}}_k$ is also
found by numerical 1-D search.

Note that, for the above derivations, we have assumed that the
value of the channel error variance $\sigma_{\Delta}^2$ is perfectly
known. In the practical systems, the channel error variance can be
estimated through an appropriate statistical approach \cite{Davidson_2007}.
Below, we also present numerical performance results corresponding
to the cases where the error variance is not perfectly known.

\section{Discussion: Computational complexity, channel state information}
In this section, we analyze computational complexity and required
feedback resources. For comparison, we also analyze those of the
gradient descent method of \cite{LeeInKyu_KuserIC}.

\subsection{Computational complexity}
We consider the number of complex multiplications as a complexity
measure. As summarized in the Table \ref{tab:complexity}, the
number of complex multiplications is proportional to the number of
iterations. The proposed method with the sum-power constraint
which has a single iteration loop is computationally the most
efficient. Whereas both the proposed method with the individual-power constraint and the gradient descent method require double
iteration loops, i.e., the outer loop for updating the sum rate
and the inner loop for adjusting the Lagrange multiplier (in the
case of the proposed method) or for updating the step size (in the
case of the gradient-based method). Calculating the gradient and
adjusting the step size require more computational resources.
According to simulation, when SNR = $10$ dB which is in the mid
SNR regime, $K = 4$, $M = N = 5$, and $d = 2$, the minimum average
numbers of iteration for the convergence of sum rate, updating the
step size of gradient method and 1-D search with bisection method
are $10$, $10$ and $10$, respectively. In accordance with these,
$I_1=10$, $I_2=10$ and $I_3=10$ are chosen. The symbols $c_N^1$,
$c_{NM}^2$ and $c_N^3$ denote the computational complexity of a
matrix inversion of $N\times N$ matrix, a singular value
decomposition of $N\times M$ matrix, and a Cholesky factorization
of $N\times N$ matrix, respectively. The corresponding values for
those variables are $\frac{2}{3}N^3$, $7NM^2+4M^3$, and
$\frac{1}{3}N^3$, respectively \cite{Golub_book}. Fig.
\ref{FIG:complexity} shows comparison when $M=N=5$ and $d=2$
\footnote{To see the effect of the number of $K$, we fixed
$M=N=5$, even though the degree of freedom (DoF) is not achievable when $K\ge 5$}. As
expected, for the same WSR values the proposed method with the sum-power constraint has the least complexity while the gradient
descent algorithm is the most computationally complex.

\subsection{The amount of required feedback information}

To find the optimized transmit precoders, each transmit node
requires feedback information. As illustrated in Table II,
feedback information is composed of CSI and coefficients for
filter updating. For a given transmission slot, CSI feedback is
required once, but the filter coefficients are updated several
times due to the iterative optimization algorithm. Although the
proposed method requires a larger amount of feedback information for the
iteratively updated coefficients such as MSE weights $\lbrace
{\bf{W}}_k\rbrace$ and receive filter coefficients $\lbrace
{\bf{U}}_k\rbrace$ than the gradient descent method does, the amount of CSI
feedback for the proposed method is smaller than for the
gradient descent method. This is because, unlike the global CSI
requirement of the gradient-based method, the proposed methods need
only local CSI. From Table II, we observe that as the network size
grows (i.e., $K$ increases) the required feedback resources for local
CSI and coefficient updating increase linearly, but those for
global CSI increases quadratically. Fig. \ref{FIG:FB_comp} clearly
shows that with $I_1=10$ the proposed
methods are advantageous in terms of required feedback resources,
especially for larger $K$. Note that, for the transmit power
adjustment, the sum-power-constrained method additionally requires
iterative update of the scalar parameter $\text{Tr}\lbrace\sum_{i\ne
k}{\bf{V}}_i^{'}{\bf{V}}_i^{'H}\rbrace$, but the size of this parameter is negligible
compared to other matrix parameters.

\section{Numerical Results}

In this section, we provide the numerical results related to the
WSR performances. The SNR for the
sum-power-constrained network,
$\text{SNR}=\frac{P_T\sigma_h^2}{K\sigma_n^2}$, and that for the
individual power constrained network,
$\text{SNR}_k=\frac{P_k\sigma_h^2}{\sigma_n^2},\forall k$, are
derived assuming $P_T = K$, $P_k = 1 \forall k$ and
$\sigma_n^2=1$, i.e., $\text{SNR}=\text{SNR}_k=\sigma_h^2$. The
results are averaged over 1000 independent trials. Fig.
\ref{FIG:WSR} shows the average WSR performance of
the proposed methods for $M = N = 5$ (when $K=4$), $M = N = 6$
(when $K=5$), and $d = 2$. For fairness, all schemes are
initialized with the right singular matrices of the intended
channels. For the sum-power constraint, we set the weights to be $\mu_1=2$ and $\mu_k
=0.25$ $(k\ne 1)$, which were chosen rather arbitrarily except that
$\mu_1$ is made considerably larger than $\mu_k$ to bring out the performance advantage of the
sum-power constraint. The performance of the sum power constraint
method should be better than that of the individual power constraint
method because the former, which is less stringent, is
able to allocate more power to the higher weighted transmitter to
maximize the WSR. When the weights are equal,
$\mu_k=1$ $\forall k$, the performance of both proposed schemes
and that of the conventional gradient descent method are nearly
identical. Note that, as explained in section IV, the proposed
methods achieve these performances with less computational
complexity and a smaller amount of feedback resources than the
gradient descent method. Compared to the performance of the MMSE
transceiver without the MSE weights \cite{Shen:2010,Heath_IA:2011}
(curves labelled "Simple MMSE"), the advantage of designed MSE
weights is clearly shown as SNR grows. Fig. \ref{FIG:WSR_CHE}
demonstrates the effectiveness of the robust design with either
transmit power constraint in presence of channel uncertainty when
$K=4$ and $\sigma_{\Delta}^2=0.1\sigma_{h}^2$. As SNR grows, the
amount of leakage interference due to CSI imperfection also
increases. This is why the performance is saturated in the high SNR
regime in Fig. \ref{FIG:WSR_CHE}. To reflect a potential error in
estimating $\sigma_{\Delta}^2$, we model the channel error
variance as $\sigma_{\Delta}^2+\sigma_{\epsilon}^2$, where
$\sigma_{\Delta}^2$ is the actual channel error variance and
$\sigma_{\epsilon}^2$ indicates over-estimation. As shown in Fig.
\ref{FIG:WSR_CHE}, at SNR $= 15$ dB at most 3 \% sum rate losses
are shown when $\sigma_{\epsilon}^2=0.1\sigma_{\Delta}^2$.
Although not shown, same results were observed for
under-estimating the channel estimation error variance.

\section{Conclusion}
In this letter, we have studied a linear transceiver design method
for the K-user MIMO interference channel. To maximize the weighted sum
rate with less computational complexity and a smaller amount of
feedback resources, the proposed transceivers are designed in the
weighted MMSE sense with suitably chosen MSE weights. Also, the
proposed transceiver design considers both the sum-power-usage
constraint and the individual-power
constraint. Through numerical simulation, we have demonstrated
that the weighed-sum-rate performances of the proposed schemes approach
that of the existing gradient descent method. The proposed methods
have clear advantage in terms of processing requirements as well
as feedback resources over the gradient-based technique. Also,
modified versions of proposed schemes have been provided for
compensating channel mismatch.


\newpage

%
%
%

 \ifCLASSOPTIONcaptionsoff
  \newpage
\fi

\bibliographystyle{IEEEtran}
\bibliography{IEEEabrv,References}


\newpage
\begin{figure}[t]
\begin{center}
\leavevmode\epsfxsize=0.7\textwidth
\epsffile{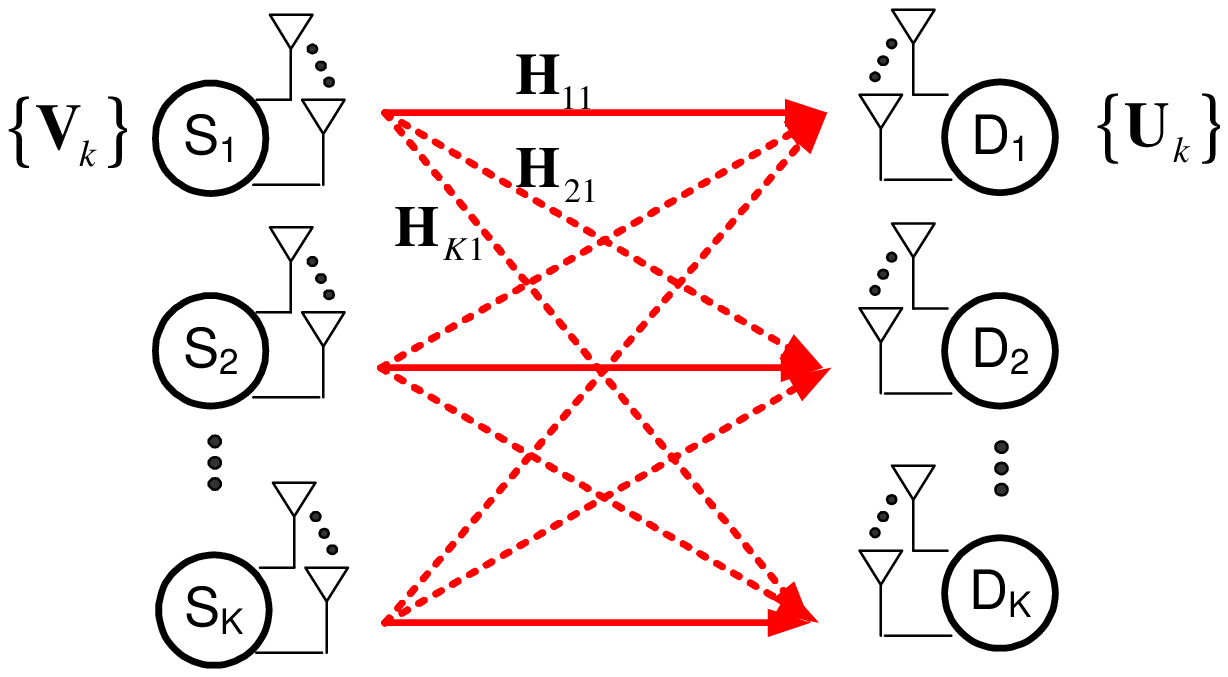} \caption{{\small{K-user
 MIMO interference channel}}}
\label{FIG:system_model}
\end{center}
\end{figure}

\newpage
\begin{figure}[t]
\begin{center}
\leavevmode\epsfxsize=0.7\textwidth
\epsffile{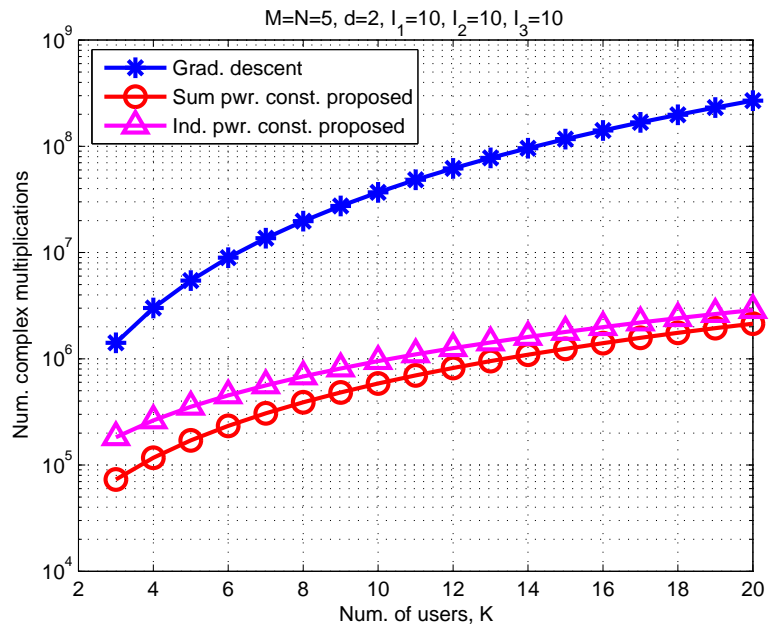} \caption{{\small{Complexity
comparison between the gradient-descent-based method and the
proposed methods}}}
\label{FIG:complexity}
\end{center}
\end{figure}

\newpage
\begin{figure}[t]
\begin{center}
\leavevmode\epsfxsize=0.7\textwidth
\epsffile{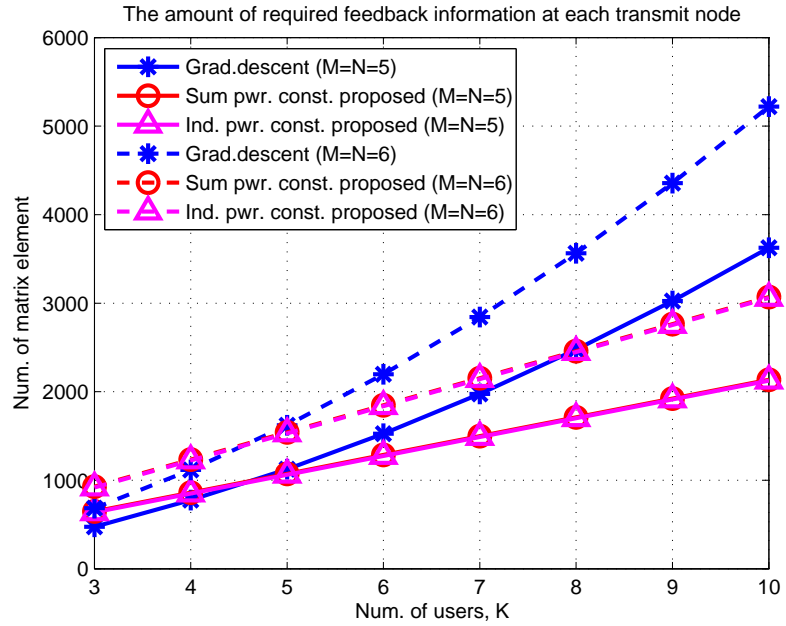} \caption{{\small{The amount of
feedback information at each source node to design precoder}}}
\label{FIG:FB_comp}
\end{center}
\end{figure}

\newpage

\begin{figure}
\centering \mbox{\subfigure[K=4, M=N=5,
d=2]{\leavevmode\epsfxsize=0.5\textwidth
\epsffile{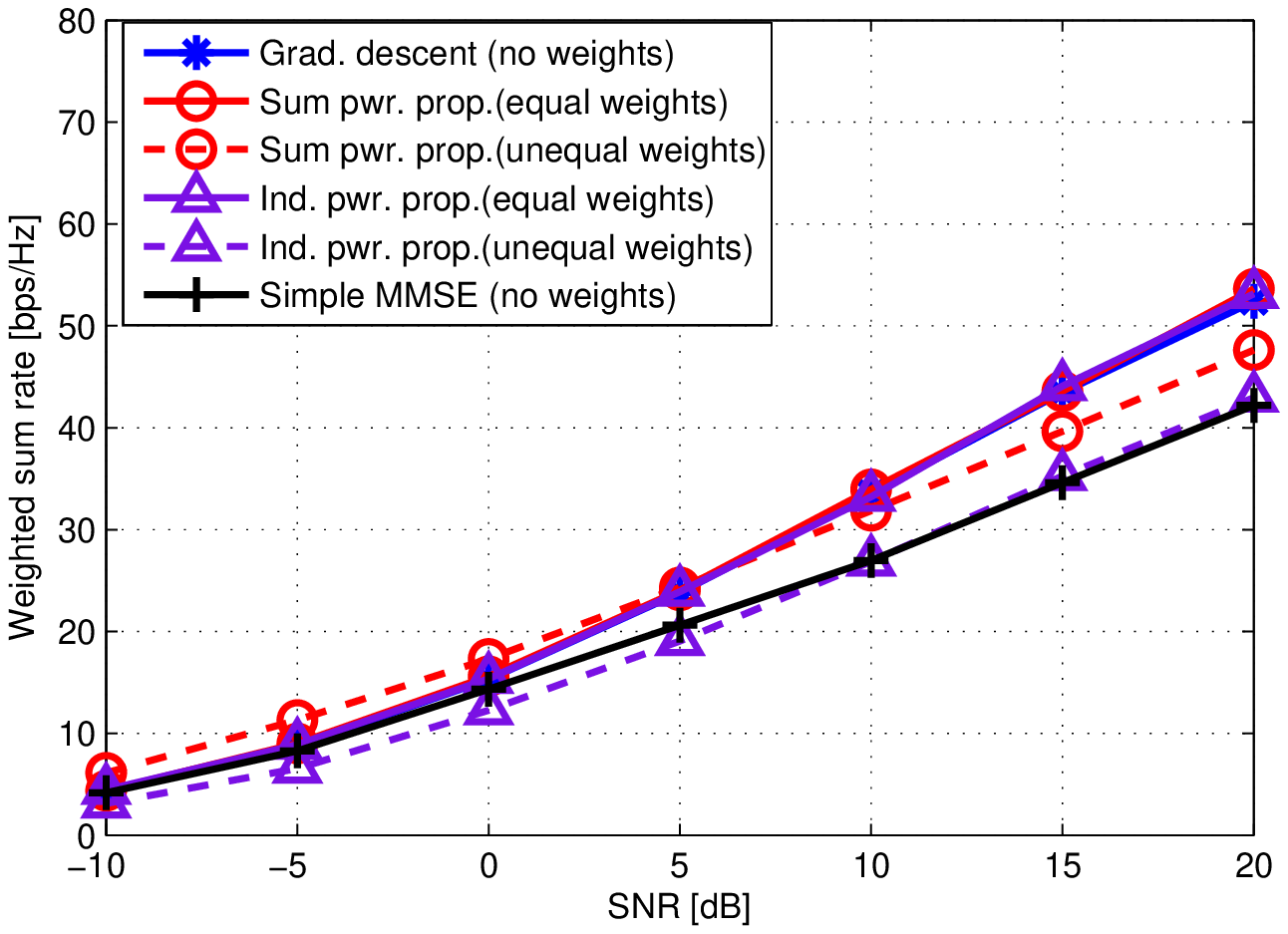}}\quad \subfigure[K=5, M=N=6,
d=2]{\leavevmode\epsfxsize=0.5\textwidth
\epsffile{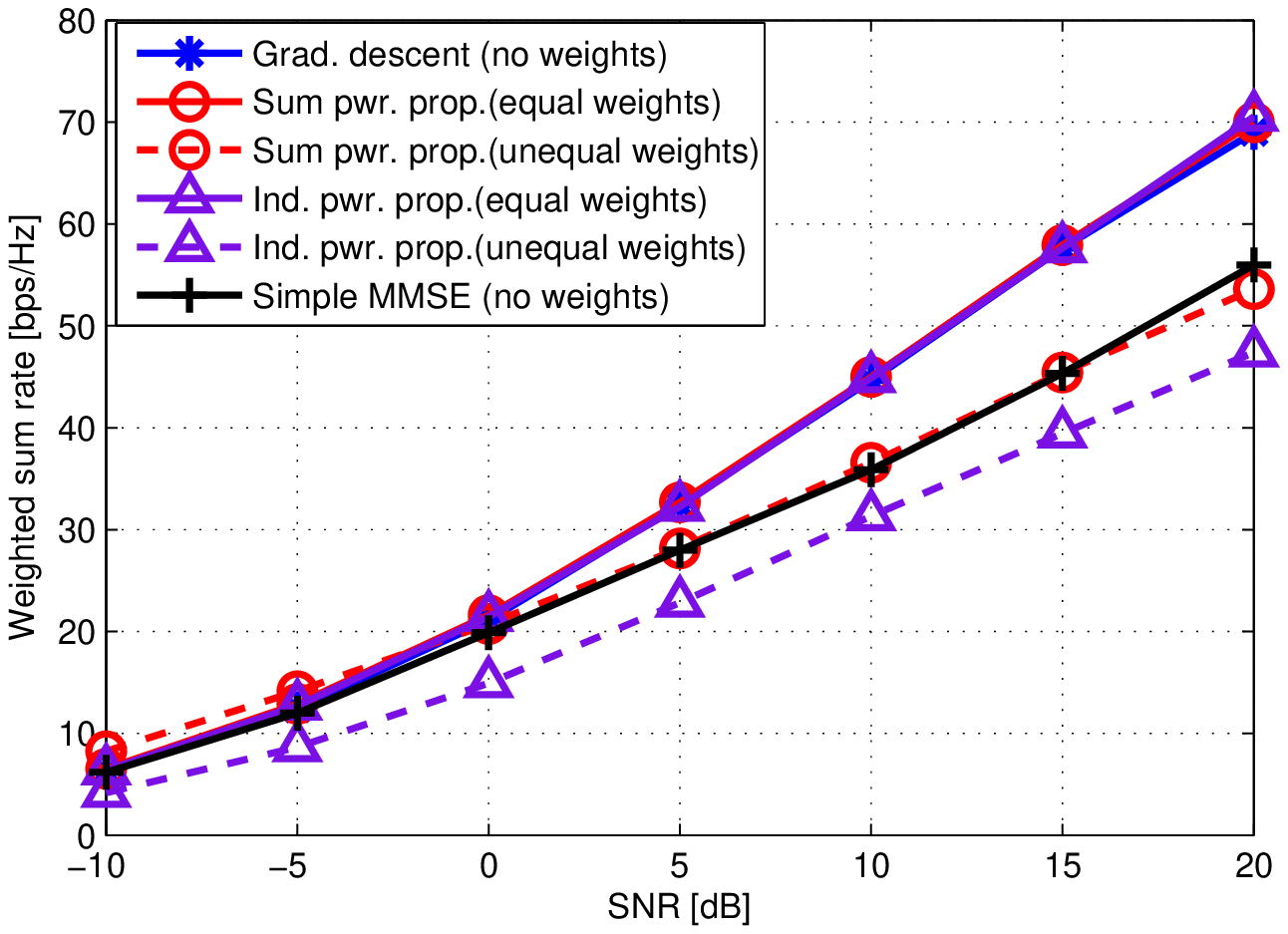}}} \caption{The weighted sum rate
performance (equal weight $\mu_k=1, \forall k$ and unequal weights
$\mu_1=2,\mu_{k'}=0.25,{k'}\ne 1$)} \label{FIG:WSR}
\end{figure}

\newpage

\begin{figure}
\centering \mbox{\subfigure[Equal weights ($\mu_k=1,\forall
k$)]{\leavevmode\epsfxsize=0.5\textwidth
\epsffile{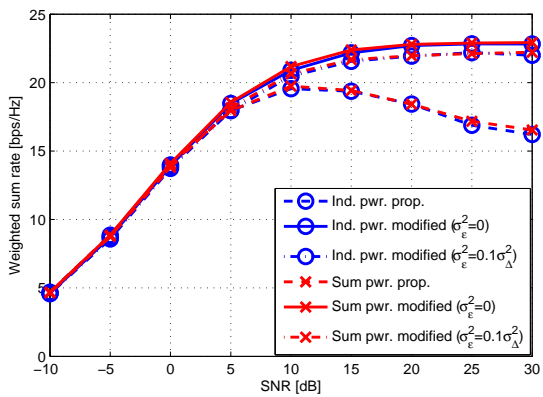}}\quad \subfigure[Unequal
weights ($\mu_1=2,\mu_k=0.25, k\ne
1$)]{\leavevmode\epsfxsize=0.5\textwidth
\epsffile{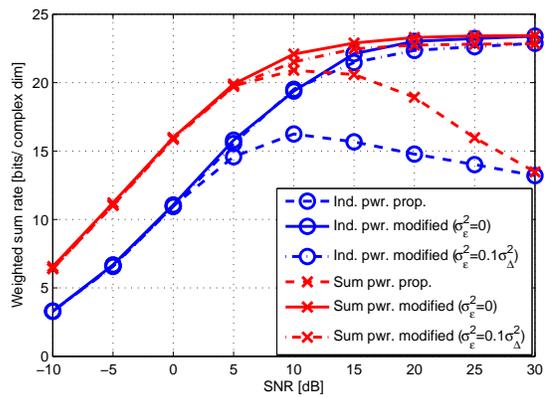}}} \caption{The weighted
sum rate performances with imperfect channel state information
($\sigma_{\Delta}^2=0.1\times \sigma_h^2, K=4)$}
\label{FIG:WSR_CHE}
\end{figure}

\newpage
\begin{table}[h!]
\centering \subtable[Description of each stage for gradient
descent method]{ \centering
\begin{tabular}{|c|c|c|}
\hline \multicolumn{2}{|c|}{STAGE} & \multicolumn{1}{c|}{Index} \\
\hline
\multicolumn{2}{|c|}{Initialization} & \multicolumn{1}{c|}{a.1} \\
\hline
 & \multicolumn{1}{|c|}{Calculating gradient}&  a.2 \\ \hhline{|~|-|-|}
Outer loop & \multicolumn{1}{|c|}{Inner loop: calculating step size}& a.3 \\
\hhline{|~|-|-|}
 & \multicolumn{1}{|c|}{Calculating sum rate}&  a.4 \\ \hline
\multicolumn{2}{|c|}{Calculating optimal precoders and decoders} & \multicolumn{1}{c|}{a.5 } \\
\hline
\end{tabular}
\label{tab:complexity:a}}

\subtable[Description of each stage for proposed methods]{
\centering
\begin{tabular}{|c|c|c|} \hline
\multicolumn{2}{|c|}{STAGE} & \multicolumn{1}{c|}{Index} \\
\hline
\multicolumn{2}{|c|}{Initialization} & \multicolumn{1}{c|}{b.1} \\
\hline
 & \multicolumn{1}{|c|}{Calculating the variance of noise and
 interference}& b.2
 \\ \hhline{|~|-|-|}
 & \multicolumn{1}{|c|}{Calculating the receive filter}& b.3
 \\ \hhline{|~|-|-|} Loop
 & \multicolumn{1}{|c|}{Calculating the error covariance matrix}& b.4
 \\ \hhline{|~|-|-|}
 & \multicolumn{1}{|c|}{Calculating the MSE weights}& b.5
 \\ \hhline{|~|-|-|}
 & \multicolumn{1}{|c|}{Calculating the transmit filter}&
 b.6-1 (for sum power constraint) \\
 & \multicolumn{1}{|c|}{(1-D search is needed for individual power constraint)}& b.6-2 (for
individual power constraint)
 \\ \hhline{|~|-|-|}
 & \multicolumn{1}{|c|}{Calculating sum rate}& b.7 \\
\hline
\end{tabular}
\label{tab:complexity:b}}

\subtable[Number of complex multiplication at each
stage]{\centering
\begin{tabular}{|c|c|} \hline
Index& Number of complex multiplication \\
\hline
 a.1 & $ K(M^2 d +1) +K(K-1)(1+2MNd+N^2d) +
K(2+2MNd+N^2d+N^3+c_{N}^{1})$ \\ \hline
 a.2 & $ I_1\Big{\lbrace} K(2K-1)(1+2MNd+N^2d) $ \\
     & $ +K(2K-1)(9+2c_N^1+2MN^2+2M^2N+2M^2d+Md^2)
 \Big{\rbrace}$ \\ \hline
 a.3 & $I_1\Big{\lbrace} KI_2(I_2+1)/2 +  KI_2\lbrace2K(K-1)(1+2MNd+N^2d) $ \\
     & $+2K(2+2MNd+N^2d+N^3+c_N^1)+K(M^2d+1)+2+Md^2  \rbrace
     \Big{\rbrace}$ \\ \hline
 a.4 & $I_1\lbrace K(M^2 d +1) +K(K-1)(1+2MNd+N^2d) +
 K(2+2MNd+N^2d+N^3+c_{N}^{1})\rbrace$\\ \hline
 a.5 & $ K(1+2Md+2M^2d+c_{Md}^{2}) + K(K-1)(2MNd+N^2d)$\\
      & $ +K(2MNd+2N^2d+4Nd^2+Md^2+d^3+c_N^1+c_d^3+c_d^2+c_{dd}^2)$\\
      \hline\hline
 b.1 & $K(M^2 d +1) +K(K-1)(1+2MNd+N^2d) +
K(2+2MNd+N^2d+N^3+c_{N}^{1})$ \\ \hline
 b.2 & $I_1 K(K-1)(1+2MNd+N^2d)$ \\ \hline
 b.3 & $I_1 K(3MNd+2N^2d+c_N^1)$  \\ \hline
 b.4 & $I_1 K(2MNd+N^2d+Nd^2+c_N^1+c_d^1) $  \\ \hline
 b.5 & $I_1 K c_d^1$  \\ \hline
 b.6-1 & $I_1 \Big{\lbrace} K(K-1)(2NMd+Md^2+M^2d) +K(Nd^2+d^3) $ \\ & $+K(3MNd+2Md^2+M^2d+1+c_M^1) +K(M^2d+Md) \Big{\rbrace}$  \\ \hline
 b.6-2 & $I_1 \Big{\lbrace} K(K-1)(2NMd+Md^2+M^2d)  +I_3 K M^2d $
\\ & $+(I_3+1)K(3MNd+2Md^2+M^2d+1+c_M^1)\Big{\rbrace}$  \\
\hline
 b.7 & $I_1\lbrace K(M^2 d +1) +K(K-1)(1+2MNd+N^2d) + K(2+2MNd+N^2d+N^3+c_{N}^{1})\rbrace $  \\ \hline
\end{tabular}
\label{tab:complexity:c}} \caption{Computational complexity
comparison} \label{tab:complexity}
\end{table}


\newpage
\begin{table}[h!]
\centering
\begin{tabular}{|c|c|c|c|c|} \hline
&\multicolumn{2}{c}{Grad. descent method}
&\multicolumn{2}{|c|}{Prop. method} \\ \cline{2-5}
 & Global CSI & Updating coefficients & Local CSI & Updating coefficients\\
\hline \multirow{2}*[-.3ex]{Feedback information} & $\lbrace
{\bf{H}}_{ij}\rbrace$ & $\lbrace {\bf{V}}_{i}\rbrace$, $(i\ne k)$
& $\lbrace {\bf{H}}_{ik}\rbrace$ & $\lbrace {\bf{U}}_{i}\rbrace$,
$\lbrace {\bf{W}}_{i}\rbrace$ (Ind. pwr.) \\ \cline{5-5} & & & &
$\lbrace {\bf{U}}_{i}\rbrace$, , $\lbrace {\bf{W}}_{i}\rbrace$,
$\sum_{i\ne k}\text{Tr}\lbrace{\bf{V}}_i{\bf{V}}_i^{'H}\rbrace$
(Sum pwr.)
\\ \hline
\multirow{2}*[-.3ex]{Matrix size} & $MNK^2$ & $Md(K-1)$ & $MNK$ &
$(Md+d^2)K$ (Ind. pwr.) \\ \cline{5-5} & & & & $(Md+d^2)K+1$ (Sum.
pwr.)
\\ \hline \multirow{2}*[-.3ex] {Feedback resource amount}  &\multicolumn{2}{c|} {$MNK^2$ + $Md(K-1) I_{1}$}   & \multicolumn{2}{c|} {$MNK$ + $(Md+d^2)K I_{1}$ (Ind. pwr.)}
\\ \cline{4-5}
&\multicolumn{2}{c|}{}& \multicolumn{2}{c|}{$MNK$ + $((Md+d^2)K+1) I_{1}$ (Sum. pwr.)} \\
\hline
\end{tabular}
\label{tab:FBamount}\caption{Summary of required feedback
information at the $k$-th transmit node, $\textsf{S}_k$,
$i,j=1\sim K$}
\end{table}

\end{document}